%% file: main.tex
\documentclass[11pt]{article}

\usepackage{amsmath}
\usepackage{amssymb}
\usepackage{amsthm}
\usepackage{tikz}
\usepackage{graphicx}
\usepackage{float}
\usepackage{algorithm}
\usepackage[noend]{algpseudocode}
\usepackage{tabularx}
\usepackage[margin=1in]{geometry}
\usepackage{arydshln}
\usepackage{multirow}
\usepackage{wrapfig}
\usepackage{url}

\usetikzlibrary{positioning,calc,shapes,arrows}
\usetikzlibrary{backgrounds}
\usetikzlibrary{matrix,shadows,arrows} 
\usetikzlibrary{decorations.pathreplacing}

\def\etal.{et\penalty50\ al.}
\theoremstyle{plain}
\newtheorem{theorem}{Theorem}[section]
\newtheorem{lemma}[theorem]{Lemma}

\newtheorem{corollary}[theorem]{Corollary}

\newtheorem{claim}[theorem]{Claim}

\theoremstyle{definition}
\newtheorem{definition}{Definition}[section]

\theoremstyle{remark}

\theoremstyle{plain}
\newtheorem*{theorem*}{Theorem}

\DeclareMathOperator{\OPT}{OPT}
\DeclareMathOperator{\ALG}{ALG}

\DeclareMathOperator{\RTPAM}{RTPAM}
\DeclareMathOperator{\Poi}{Poi}
\DeclareMathOperator{\Bin}{Bin}
\DeclareMathOperator{\Var}{Var}
\DeclareMathOperator{\Ber}{Ber}

\title{Greedy Bipartite Matching in Random Type Poisson Arrival Model}
\author{Allan Borodin\thanks{Research is supported by NSERC.} \\ University of Toronto \\ \textsf{bor@cs.toronto.edu}  \and 
Christodoulos Karavasilis\footnotemark[1] \\ University of Toronto \\ \textsf{ckar@cs.toronto.edu} \and 
Denis Pankratov\footnotemark[1] \\ University of Toronto \\ \textsf{denisp@cs.toronto.edu}}
\date{\today}

\begin{document}

\maketitle

\begin{abstract}
We introduce a new random input model for bipartite matching which we call the Random Type Poisson Arrival Model. Just like in the known i.i.d.  model (introduced by Feldman et al.~\cite{FeldmanMMM2009}), online nodes have types in our model. In contrast to the adversarial types studied in the known i.i.d.  model, following the random graphs studied in    Mastin and Jaillet~\cite{MastJailGnnp}, in our model each type graph is generated randomly by including each offline node in the neighborhood of an online node with probability $c/n$ independently. In our model, nodes of the same type appear consecutively in the input and the number of times each  type node appears is distributed according to the Poisson distribution with parameter 1. We analyze the performance of the simple greedy algorithm under this input model. The performance is controlled by the parameter $c$ and we are able to exactly characterize the competitive ratio for the regimes $c = o(1)$ and $c = \omega(1)$. We also provide a precise bound on the expected size of the matching in the remaining regime of constant $c$. We compare our results to the previous work of Mastin and Jaillet who analyzed the simple greedy algorithm in the $G_{n,n,p}$ model where each online node type occurs exactly once. We essentially show that the approach of Mastin and Jaillet can be extended to work for the Random Type Poisson Arrival Model, although several nontrivial technical challenges need to be overcome. Intuitively, one can view the Random Type Poisson Arrival Model as the $G_{n,n,p}$ model with less randomness; that is, instead of each online node having a new type, each online node has a chance of repeating the previous type.
\end{abstract}

\section{Introduction}
\label{sec:intro}
\input{intro}

\section{Preliminaries}
\label{sec:prelim}
\input{prelim}

\section{Greedy Matching in Random Type Poisson Arrival Model}
\label{sec:greedy_poisson}
\input{greedy_poisson}

\section{Conclusion}
\label{sec:conclusion}
\input{conclusion}

\bibliographystyle{plain}
\bibliography{thesis}

\appendix

\section{Two Technical Lemmas}
\label{sec:tech_lemmas}
\input{tech_lemmas}

\section{Figures}
\label{sec:figures}
\input{figures.tex}

\clearpage
\section{The 3-Step Version of $\RTPAM$}
\label{sec:3-step}
\input{three_step.tex}

\end{document}

%% file: intro.tex
Online bipartite matching is a problem with a wide variety of applications and has received significant attention after the seminal work of Karp, Vazirani and Vazirani~\cite{KarpVV90} who showed an optimal $(1-1/e)$ randomized algorithm for the unweighted case in the adversarial input model. Applications such as  internet advertising (see the excellent survey by Mehta~\cite{MehtaSurvey} and references therein) and job allocation have given rise to problems that are naturally described as bipartite matching problems and this has been a strong motivation for developing better algorithms. Most of the work in online bipartite matching is with respect to a vertex arrival model where vertices on one side of the bipartite graph are known to the algorithm in advance and vertices on the other side are revealed online along with their adjacent vertices. 

Adversarial models are, of course, pessimistic in terms of performance results. For many applications, more realistic assumptions will lead to significantly improved results. In this regard, various stochastic input models have been proposed and analyzed. These include the random order model, 
known and unknown i.i.d.  distribution models, and the Erd\H{o}s-R\'{e}nyi random graph model. We propose a new model (see Section~\ref{sec:preliminaries}) that is closely related to both the i.i.d.  model and the Erd\H{o}s-R\'{e}nyi model. The motivation for our model is that in applications, one can often observe some bursts of identical inputs. One can think of this as a very restricted form of a Markov chain model. 

Whereas previous work in the i.i.d.  models take advantage of the input distributions in terms of the decisions being made (i.e. as to how to match the online vertices), we follow the work of Mastin and Jaillet~\cite{MastJailGnnp} who study the performance of the simplest greedy algorithm (which we will call \textsc{Greedy}) that always matches (when possible) an online vertex to the first (in some fixed ordering of the offline vertices) unmatched adjacent neighbor. We note that this  simple greedy algorithm does not utilize any information about the nature of the input sequence (including the number of online vertices). 

We know that theoretically (i.e. in terms of the expected approximation ratio), that knowledge of the distribution and/or randomization in the algorithm will allow for online algorithms with improved performance. However, in simulations with respect to various stochastic models, we find that the simplest  deterministic greedy algorithm is  competitive with more specialized algorithms \cite{BKP2018}. We also note that when a type graph is chosen adversarially in the i.i.d. model, Goel and Mehta \cite{GoelM2008} show that the approximation ratio of  \textsc{Greedy} is precisely $1-1/e \approx .632$.   

Mastin and Jaillet show that the  approximation ratio of  \textsc{Greedy} in the Erd\H{o}s-R\'{e}nyi model is at least $0.837$. Intuitively, our model introduces correlations between consecutive online nodes. This causes significant technical difficulties in the analysis that we are able to overcome. In addition, the correlations do not seem favourable to the  \textsc{Greedy} algorithm, and it is natural to expect \textsc{Greedy} to have worse performance in our model than in the Erd\H{o}s-R\'{e}nyi model. Our results confirm this intuition. Our input model has a parameter $c$ that controls the expected degree of online nodes. Similar to the work of Mastin and Jaillet~\cite{MastJailGnnp}, we analyze the performance of \textsc{Greedy} in three different regimes of $c$: $c = o(1)$, $c = \omega(1)$, and constant $c$. We compute the exact asymptotic fraction of offline nodes matched by \textsc{Greedy} in each of these regimes. We also show how to derive an upper bound on the size of a maximum matching in our model from the existing upper bounds on the size of a maximum matching in the Erd\H{o}s-R\'{e}nyi model. Combining the two results, we derive a lower bound on the approximation ratio of \textsc{Greedy} in all regimes of $c$. Minimizing this lower bound over $c$, we find that \textsc{Greedy} has approximation ratio $\ge 0.715$ for all values of $c$ in our model. 

\emph{Organization.} The rest of the paper is organized as follows. Section~\ref{sec:prelim} describes our model and different views of it, as well as some background material needed for the rest of the paper. Section~\ref{sec:greedy_poisson} is the technical core of the paper and contains the analysis of \textsc{Greedy} in the three regimes of $c$: $c = o(1)$, $c = \omega(1)$, and $c = \Theta(1)$. Section~\ref{sec:conclusion} concludes the paper with a brief discussion of various input models and some open questions.

%% file: prelim.tex
\label{sec:preliminaries}
We shall consider bipartite graphs $G=(U, V, E)$ in the vertex arrival model. The vertices in $U$ are referred to as the offline vertices and are known to an algorithm in advance. The vertices in $V$ are referred to as online vertices and arrive one at a time. When a vertex from $V$ is revealed, all its neighbors in $U$ are revealed as well. A simple greedy algorithm matches the arriving vertex $v \in V$ to the first available neighbor in $U$ (if there is one, according to some fixed ordering). We shall denote this algorithm by \textsc{Greedy}. We consider the behavior of \textsc{Greedy} on specific families of random graphs generated by, what we call, the Random Type Poisson Arrival Model. This model has two parameters: $n \in \mathbb{N}$ which is equal to $|U|$ and intuitively measures the size of the instance, and $c \in \mathbb{R}_{\ge 0}$ which controls the density of edges. The random graph in our model is generated as follows: the offline nodes $U$ are set to be $[n]$. Online nodes are generated iteratively and randomly using the following process. For each $i \in [n]$ generate a random type $i$ by including each offline node $j \in [n]$ in the neighborhood of the type with probability $c/n$ independently. Then we  sample $Z_i$ from the Poisson distribution with parameter 1 and generate $Z_i$ online nodes of type $i$, and continue. We shall denote a graph $G$ distributed according to the Random Type Poisson Arrival Model with parameters $n$ and $c$ as $G \sim \RTPAM(n, c)$. Our model can be viewed in different ways. We refer to the current view of the $\RTPAM(n,c)$ model as the ``one-step view.'' Next we describe another view.

The bipartite Erd\H{o}s-R\'{e}nyi model is denoted by $G_{n,n,p}$. The random graph $G = (U, V, E)$ in the $G_{n,n,p}$ model is generated as follows: $U = V = [n]$, and each edge $\{i,j\}$ is included in $G$ with probability $p$ independently. Our $\RTPAM(n,c)$ model can be alternatively viewed as a two-step process. At first, we generate a \emph{type graph} $\widehat{G}=(U,V,E)$ from the  $G_{n,n,p}$ distribution where $p = c/n$. For each type $i \in V$ we sample $Z_i \sim \Poi(1)$. Secondly, an \emph{instance graph} is created by keeping the same $U$ as in the type graph, and replicating each type $i$ node $Z_i$ times (note that this means removing type $i$ when $Z_i = 0$). In the rest of the paper, we shall freely switch between the two different views of the $\RTPAM(n,c)$ model. Thus, we shall occasionally refer to the type graph of the $\RTPAM(n,c)$ graph. We refer to this alternative view of the $\RTPAM(n,c)$ model as the ``two-step view.''

Intuitively, in the one-step view, $Z_i$ is drawn from the Poisson distribution  in an online fashion whereas in the two-step view, $Z_i$ is drawn initially. It should be clear that these views do not change the model. However, our proofs are facilitated by having these different views. 

We shall measure the performance of an algorithm in two ways: in terms of the asymptotic approximation ratio, and in terms of the fraction of the matched offline nodes.

\begin{definition}
Let $\ALG$ be a deterministic online algorithm solving the bipartite matching problem over random graphs $G_n$ parameterized by the input size $n$. We write $\ALG(G_n)$ to denote the size of the matching (random variable) that is constructed by running $\ALG$ on $G_n$. We write $\OPT(G_n)$ to denote the size of a maximum matching in $G_n$. \emph{The asymptotic approximation ratio} of $\ALG$ with respect to $G_n$ is defined as:
\[\rho(\ALG, G_n) = \liminf_{n \rightarrow \infty} \frac{\mathbb{E}(\ALG(G_n))}{\mathbb{E}(\OPT(G_n))}.\]
\emph{The fraction of matched offline nodes} of $\ALG$ with respect to $G_n$ is defined as:
\[\mu(\ALG, G_n) = \liminf_{n \rightarrow \infty} \frac{\mathbb{E}(\ALG(G_n))}{n}.\]
\end{definition}

We shall use the following notation for the well-known distributions:

\begin{definition}
\begin{itemize}
    \item $\Poi(\lambda)$ --- the Poisson distribution with parameter $\lambda$,
    \item $\Bin(n, p)$ --- the Binomial distribution with parameters $n \in \mathbb{Z}_{\ge 0}$ (number of trials) and $p \in [0,1]$ (probability of success)
    \item $\Ber(p)$ --- the Bernoulli distribution with parameter $p$.
\end{itemize}
We shall also write $\Poi(\lambda)$, $\Bin(n,p)$, etc., as a placeholder for a random variable distributed according to the corresponding distribution. This is done to simplify the notation when the name of the random variable is not important and only the parameters of the distribution are of interest.
\end{definition}

\begin{definition}
We write $\mathbb{I}(E)$ to denote the indicator random variable for an event $E$.
\end{definition}

In the paper, we show how to compute $\mu(\textsc{Greedy}, \RTPAM(n,c))$ exactly. In order to provide a bound on the asymptotic approximation ratio of \textsc{Greedy} we also need to know the size of a maximum matching in $\RTPAM(n,c)$. The size of a maximum matching is not known even in $G_{n,n,c/n}$ model, but there is a known upper bound due to Bollob\'{a}s and Brightwell that we will be able to use to derive a lower bound on the asymptotic approximation ratio of \textsc{Greedy} in the $\RTPAM(n,c)$ model.

\begin{theorem}[\cite{BollobasB95}]
\label{thm:match_gnnp_ub}
\[\frac{\mathbb{E}(\OPT(G_{n,n,c/n}))}{n} \le 2 - \frac{\gamma^* + \gamma_* + \gamma^* \gamma_*}{c} + o(1),\]
where $\gamma_*$ is the smallest solution to the equation $x = c\exp(-c \exp(-x))$ and $\gamma^*=c \exp(-\gamma_*)$. See Figure~\ref{fig:match_gnnp_ub} for the shape of this upper bound as a function of $c$.
\end{theorem}

We shall later compare our results with the following result giving the performance of \textsc{Greedy} in the $G_{n,n,c/n}$ model due to Mastin and Jaillet.

\begin{theorem}[\cite{MastJailGnnp}]
\label{thm:greedy_gnnp}
For $c \in \mathbb{R}_{> 0}$ we have
\[ \mu(\textsc{Greedy}, G_{n,n,c/n}) = 1- \frac{\log(2-e^{-c})}{c}.\]
\end{theorem}

%% file: greedy_poisson.tex
\subsection{The Regime of $c = o(1)$}

In this section we show that in expectation \textsc{Greedy} finds an almost-maximum matching in $\RTPAM(n, c)$ when $c = o(1)$. The high level idea is to consider the two-step view of $\RTPAM(n,c)$ and observe that when $c = o(1)$ most of the type graph consists of \emph{isolated} edges --- edges with both endpoints of degree 1. Thus, no matter how many times an online node corresponding to an isolated edge is generated in the instance graph, both \textsc{Greedy} and $\OPT$ can match it exactly once (given that it is generated at all). The expected number of the non-isolated edges is of smaller order of magnitude than the expected number of isolated edges and can be ignored for the purpose of computing the asymptotic approximation ratio.

\begin{theorem}
\label{thm:little_c_regime}
Let $c = o(1)$ then we have:
\[\rho (\textsc{Greedy}, \RTPAM(n,c)) = 1.\]
\end{theorem}
\begin{proof}{(similar to the proof of Lemma 2 in \cite{MastJailGnnp})}

Set $p = c/n = o(1/n)$, and consider the two-step view of $\RTPAM(n,c)$. The probability that an edge between $i$ and $j$ appears and is isolated in \emph{the type graph} is:

\[\Pr(\{i,j\} \text{ is isolated in the type graph })=p(1-p)^{2n-2}.\]

Observe that if type $j$ is generated at least once in \emph{the instance graph}, i.e., $Z_j \ge 1$, and $\{i,j\}$ is an isolated edge in $\emph{the type graph}$ then \textsc{Greedy} will include $\{i,j\}$ in the matching exactly once. In addition, observe that the events $Z_j \ge 1$ and ``$\{i,j\}$ is isolated in the type graph'' are independent. Let $W_{i,j}$ be a random variable indicating whether $\{i,j\}$ is included by \textsc{Greedy}. Then, we have:

\[ \Pr(W_{i,j} = 1) \ge \Pr(Z_j \ge 1) \Pr(\{i,j\} \text{ is isolated in the type graph}) =  \left(1-\frac{1}{e}\right )p(1-p)^{2n-2}.\]

Let $M$ be the matching produced by \textsc{Greedy}. We have $|M| = \sum_{i,j} W_{i,j}$. It follows that:

\[\mathbb{E}(|M|) = \mathbb{E}\left(\sum_{i,j} W_{i,j}\right) = \sum_{i,j} \Pr(W_{i,j} = 1) \geq n^{2} \left(1-\frac{1}{e}\right)p(1-p)^{2n-2}.\]

Let $Q_j$ denote the number of neighbors of a node of type $j$. Observe that in the instance graph the nodes corresponding to type $j$ can be matched in an optimal matching at most $Q_j \mathbb{I}(Z_j \ge 1)$ times. Note that $Q_j$ and $Z_j$ are independent. Let $M^*$ denote a maximum matching in the instance graph. Then we have

\[ \mathbb{E}(|M^*|) \le \sum_j \mathbb{E}(Q_j \mathbb{I}(Z_j \ge 1)) = \sum_j np (1-1/e) = n^2 p (1-1/e).\]

Combining this with the above lower bound on $\mathbb{E}(|M|)$ we get

\[\frac{\mathbb{E}(|M|)}{\mathbb{E}(|M^{*}|)}\geq (1-p)^{2n-2}=\left(e^{-o(\frac{1}{n})}\right)^{2n-2}=e^{-o(1)} \xrightarrow[n \to \infty]{} 1.\]
\end{proof}

\subsection{The Regime of Constant $c$}

Fix a constant $c \in \mathbb{R}_{>0}$. We can view \textsc{Greedy} as constructing the matching in rounds. Consider the one-step view of $\RTPAM(n, c)$.  During round $i$, a new type is created and $Z_i$ nodes corresponding to that type are generated. We let $Y_i$  denote the number of online nodes matched by \textsc{Greedy} by the beginning of round $i$. We also let $X_i$ denote the number of neighbors of type $i$ that were \emph{not} matched in any of the earlier rounds. In this section, we show how to compute the asymptotic fraction of matched offline nodes by \textsc{Greedy} \emph{exactly}. More specifically, we derive an asymptotically accurate (implicit) expression for $Y_n$ and show how to compute it for each value of $c$. In addition, we show that existing upper bounds on the maximum matching in the $G_{n,n,c/n}$ model carry over to the $\RTPAM(n,c)$ model. This allows us to derive lower bounds on the \emph{competitive ratio} of \textsc{Greedy} in $\RTPAM(n,c)$.

\textit{High level idea.} We use the method of partial differential equations (see, e.g., \cite{Wormald99,Wormald1995,Kurtz70}) to derive the asymptotic behavior of $Y_n$. The goal is to write the expression $\mathbb{E}(Y_{i+1} - Y_i \mid Y_i)$ in terms of $Y_i/n$, i.e., $\mathbb{E}(Y_{i+1}-Y_i \mid Y_i) = f_c(Y_i/n)$ for some ``simple'' function $f_c$. This gives us a difference equation for $Y_i$. Now, pretend that there is a function $g_c: [0,1] \rightarrow[0,1]$ that gives a good approximation to $Y_i$, i.e.,  $g_c(t) \approx Y_{tn}/n$ for $t \in [0,1]$. Consider a syntactic replacement of the difference equation for $Y_i$ with a differential equation for $g_c$, i.e., $g_c'(t) = f_c(g_c(t))$. In addition, set the correct initial value condition $g_c(0) = Y_0 = 0$. The differential equation method allows us to conclude that under a mild condition on $f_c$ (namely, being Lipschitz), the solution $g_c$ is unique and asymptotically converges to $Y_n/n$, i.e., $Y_n = g_c(1) n + o(n)$. In particular, we have $\mu(\textsc{Greedy}, \RTPAM(n,c)) = g_c(1)$. In our setting, we will see that it is not clear how to write $\mathbb{E}(Y_{i+1} - Y_i \mid Y_i)$ as a function of $Y_i/n$.  It turns out that the method still works as long as $\mathbb{E}(Y_{i+1}-Y_i \mid Y_i)$ is close to $f_c(Y_i/n)$ in the following sense: $\lim_{n \rightarrow \infty} | \mathbb{E}(Y_{i+1}-Y_i \mid Y_i) - f_c(Y_i/n)| = 0.$ This is precisely what we do in this section.

The number of nodes matched by \textsc{Greedy} in round $i$ is exactly equal to $\min(X_i, Z_i)$. By the definition of the $\RTPAM(n, c)$ model, we have $X_i \sim \Bin(n -Y_i, c/n)$. Therefore, we have 
\begin{equation}
\label{eq:main_eq}
    \mathbb{E}(Y_{i+1}-Y_i \mid Y_i) = \mathbb{E}(\min(\Bin(n-Y_i, c/n), \Poi(1)) \mid Y_i).
\end{equation} 
Unfortunately, as mentioned above this expectation does not seem to have a nice form and we do not know how to set up an associated differential equation. Instead, we shall approximate the difference equation $\mathbb{E}(Y_{i+1}-Y_i \mid Y_i)$ by another expression that is easier to handle. Those familiar with the Poisson limit theorem will immediately recognize the following as the most natural choice:

\[\mathbb{E}(\min(\Bin(n-Y_i, c/n), \Poi(1)) \mid Y_i) \approx \mathbb{E}(\min(\Poi(c(1-Y_i/n)), \Poi(1))\mid Y_i).\]

We need to analyze how accurate this approximation is, but first we show how to derive a relatively simple expression for the right hand side. Define $h(x) := \mathbb{E}(\min(\Poi(x), \Poi(1)))$. Although the function $h(x)$ does not have a closed-form expression in terms of widely-known functions such as $\sin, \cos, \exp,$ etc., it does have a closed-form expression in terms of the modified Bessel functions of the first kind and the Marcum's Q functions:

\begin{definition}
The modified Bessel functions of the first kind are defined as follows:
\[I_k(x) = \sum_{i = 0}^\infty \frac{1}{i! \Gamma(i+k+1)} \left( \frac{x}{2} \right)^{2i+k}.\]
For an integer $k \ge 0$, it becomes $I_k(x) = \sum_{i = 0}^\infty \frac{1}{i! (i+k)!} \left( \frac{x}{2} \right)^{2i+k}.$ We note that the modified Bessel functions have the following symmetry property: $I_k(x) = I_{-k}(x).$
\end{definition}

\begin{definition}
Marcum's Q function is defined as follows:
\[ Q_n(a,b) = \exp\left(- \frac{a^2+b^2}{2} \right) \sum_{k = 1 -n}^\infty \left(\frac{a}{b} \right)^k I_k(ab).\]
\end{definition}

Now, the closed-form expression for $h(x)$ can be derived from the answer on Stats Stackexchange~\cite{Skellam}. For completeness, we reproduce the derivation in Appendix~\ref{sec:tech_lemmas}.

\begin{lemma}[\cite{Skellam}]
\label{lem:h_expr}
For $x > 0$ we have
\[ h(x) = \frac{1+x - 2 e^{-x-1}\left(I_0(2 \sqrt{x}) + \sqrt{x}I_1(2\sqrt{x})\right) - (1-x)\left(1-2Q_1(\sqrt{2 x}, \sqrt{2})\right)}{2}.\]
\end{lemma}

Thus, in terms of our overview we have $f_c(x) = h(c(1-x))$, because we hope to show that 

\[\mathbb{E}(\min(\Bin(n-Y_i, c/n), \Poi(1)) \mid Y_i) \approx \mathbb{E}(\min(\Poi(c(1-Y_i/n)), \Poi(1))\mid Y_i) = h(c(1-Y_i/n)).\]

To apply the method of differential equations we need to analyze the above approximation and show that $f_c$ is Lipschitz. We start by analyzing how good the approximation is.

\begin{lemma}
\label{lem:key_lem_1}
\[\lim_{n\rightarrow \infty} \bigl| \mathbb{E}(\min(\Bin(n-Y_i, c/n), \Poi(1)) \mid Y_i) - h(c(1-Y_i/n)) \bigr| = 0.\]
\end{lemma}

\begin{proof}
We introduce the following useful notation: $N = n - Y_i$, $p = c/n$, and $\lambda = c(1-Y_i/n)$. Moreover, we let $b_r = \Pr(\Bin(N,p)=r)$ and $p_r(a) = \Pr(\Poi(a) = r)$. Let $W_i \sim \Poi(c(1-Y_i/n))$. Then the statement of the lemma can be translated into
\[\bigl| \mathbb{E}(\min(X_i, Z_i) - \min(W_i, Z_i)\mid Y_i)\bigr|\]
The expectation is just a big sum. Let's consider individual terms  and their contributions to the overall sum.
\begin{enumerate}
\item The contribution of $X_i = k, W_i = j, Z_i = \ell$ when $j < k < \ell$ is $k-j$.
\item The contribution of $X_i = j, W_i = k, Z_i = \ell$ when $j < k < \ell$ is $j-k = -(k-j)$.
\item The contribution of $X_i = k, W_i = j, Z_i = \ell$ when $j < \ell$ and $k \ge \ell$ is $\ell - j$.
\item The contribution of $X_i = j, W_i = k, Z_i = \ell$ when $j < \ell$ and $k \ge \ell$ is $j - \ell$.
\item The contribution when $j = k$ is $0$.
\end{enumerate} 

We pair up terms corresponding to (1) with (2) and terms corresponding to (3) with (4). Define
\[ S_1 = \sum_{k,j,\ell = 0}^\infty \mathbb{I}(j < k < \ell) \left( (k-j) b_k p_j(\lambda) p_\ell(1) + (j-k) b_j p_k(\lambda) p_\ell(1) \right),\]
and 
\[ S_2 = \sum_{k,j,\ell = 0}^\infty \mathbb{I}(j < \ell, k \ge \ell) \left( (\ell-j)b_k p_j(\lambda)p_\ell(1) + (j-\ell)b_j p_k(\lambda)p_\ell(1)\right). \]

Thus, we get that
\[\bigl| \mathbb{E}(\min(X_i, Z_i) - \min(W_i, Z_i)\mid Y_i)\bigr| = |S_1 + S_2| \le |S_1| + |S_2|.\]

We will show that $\lim_{n \rightarrow \infty} |S_1| = 0$. Similar argument implies that $\lim_{n \rightarrow \infty} |S_2| = 0.$

We have
\begin{align*}
    S_1 &= \sum_{k,j,\ell = 0}^\infty \mathbb{I}(j < k < \ell) p_\ell(1) (k-j) (b_k p_j(\lambda) - b_j p_k(\lambda))\\
    &= \sum_{k,j,\ell = 0} ^\infty \mathbb{I}(j<k<\ell) p_\ell(1) (k-j) ((b_k-p_k(\lambda))p_j(\lambda) - (b_j - p_j(\lambda))p_k(\lambda))\\
    &= S_1^1 + S_1^2,
\end{align*} 
where $S_1^1 = \sum_{k,j,\ell = 0} ^\infty \mathbb{I}(j<k<\ell) p_\ell(1) (k-j) (b_k-p_k(\lambda))p_j(\lambda)$ and $S_1^2 = S_1 - S_1^1$. Again, $|S_1| \le |S_1^1| + |S_1^2|$. We will show that $\lim_{n \rightarrow \infty} |S_1^1| = 0$, and a similar argument implies that the same holds for $|S_1^2|$. We have

\[|S_1^1| \le \sum_{k = 0}^\infty \left( \sum_{j,\ell = 0}^\infty \mathbb{I}(j<k<\ell)  (k-j) p_\ell(1) p_j(\lambda) \right)|b_k-p_k(\lambda)|.\]

By \cite{Simons1971}, we have $\lim_{n\rightarrow \infty}|S_1^1| = 0$ if and only if

\[\sum_{k,j,\ell = 0}^\infty \mathbb{I}(j<k<\ell)  (k-j) p_\ell(1) p_j(\lambda) p_k(\lambda) = O(1).\]
Lastly, we have

\begin{align*}
   & \sum_{k,j,\ell = 0}^\infty \mathbb{I}(j<k<\ell)  (k-j) p_\ell(1) p_j(\lambda) p_k(\lambda) \le \sum_{k,j,\ell = 0}^\infty \mathbb{I}(j<k<\ell) k p_\ell(1) p_j(\lambda) p_k(\lambda)\\
   &\le \sum_{k,j,\ell = 0}^\infty p_\ell(1) p_j(\lambda) k p_k(\lambda) \le \lambda \xrightarrow[n \to \infty]{} c =  O(1).
\end{align*}

\end{proof}

Due to space considerations we prove that $f_c$ is Lipschitz in Appendix~\ref{sec:tech_lemmas}.

\begin{lemma}
\label{lem:key_lem_2}
The function $f_c(x)$ is Lipschitz on $[0,1]$.
\end{lemma}

Finally, we have all the necessary ingredients to prove the main theorem of this section. Although this theorem does not give an explicit closed-form expression for $\mu(\textsc{Greedy}, \RTPAM(n,c))$, it gives a simple way to evaluate it numerically for any value of $c > 0$.

\begin{theorem}
\label{thm:greedy_lb}
\[ \mu(\textsc{Greedy}, \RTPAM(n,c)) = g_c(1), \]
where $g_c$ is a solution to the following differential equation:
\begin{align*}
    g_c'(t) &= h(c(1-g_c(t))), \\
    g_c(0) &= 0, 
\end{align*}
and $h$ is given in Lemma~\ref{lem:h_expr}.
\end{theorem}
\begin{proof}
This is a direct application of Wormald's theorem (Theorem 5.1 in~\cite{Wormald99}) using Lemmas~\ref{lem:key_lem_1} and~\ref{lem:key_lem_2}.
\end{proof}

Next, we show two upper bounds on $\mathbb{E}(\OPT(\RTPAM(n,c)))$. The minimum of the two will be used to compute the approximation ratio of \textsc{Greedy}.

\begin{theorem}
\label{thm:match_ub}
    For all $c \in \mathbb{R}_{>0}$ we have
   \[  \mathbb{E}(\OPT(\RTPAM(n,c))) \le \min\left(n c \left(1 - \frac{1}{e} \right), \left(2 - \frac{\gamma^* + \gamma_* + \gamma^* \gamma_*}{c}\right) n + o(n) \right),\]
   where $\gamma_*$ is the smallest solution to the equation $x = c\exp(-c \exp(-x))$ and $\gamma^*=c \exp(-\gamma_*)$.
\end{theorem}
\begin{proof}
The first argument in the minimum follows from the proof of Theorem~\ref{thm:little_c_regime}. Let $Q_j$ denote the number of neighbors of a node of type $j$. Observe that the number of nodes of type $j$ participating in any matching is bounded above by $Q_j \mathbb{I}(Z_j \ge 1)$. Using the fact that $Q_j$ and $\mathbb{I}(Z_j \ge 1)$ are independent, taking the expectation and summing over all $j$ results in the upper bound of $nc\left(1 - \frac{1}{e}\right)$.

The second argument in the minimum follows from the observation 
\begin{equation}
\label{eq:match_rpam_vs_gnnp}
    \mathbb{E}(\OPT(\RTPAM(n,c))) \le \mathbb{E}(\OPT(G_{n,n,c/n}))
\end{equation} 
and Theorem~\ref{thm:match_gnnp_ub}. Let $\alpha(G)$ denote the independence number of the graph. By K\H{o}nig's theorem it suffices to prove that
\[ \mathbb{E}(\alpha(\RTPAM(n,c))) \ge \mathbb{E}(\alpha(G_{n,n,c/n})).\]
Consider the two-step view of $\RTPAM(n,c)$. In the first step, a type graph $\widehat{G} = (U,V,E)$ is generated from the distribution $G_{n,n,c/n}$. Let $S$ be a largest independent set in the type graph. Write $S = S_1 \cup S_2$, where $S_1 = S \cap U$ consists of offline nodes and $S_2 = S \cap V$ consists of online types. In the \emph{instance} graph all nodes from $S_1$ together with all nodes with types from $S_2$ will form an independent set. In other words, even if a node of a given type is repeated multiple times from $S_2$, it can be safely included in an independent set. Thus, we have 
\begin{align*}
    \mathbb{E}(\alpha(\RTPAM(n,c)) \mid S_1, S_2) &\ge \mathbb{E}\left(|S_1|+\sum_{j \in S_2} Z_j \mid S_1, S_2\right)\\
    &= |S_1| + \sum_{j \in S_2} \mathbb{E}(Z_j \mid S_1, S_2) = |S_1| + |S_2| = |S|,
\end{align*} 
where the last equality follows because $Z_j$ is independent of $S_1, S_2$. Taking the expectation over $S$ proves $\mathbb{E}(\alpha(\RTPAM(n,c))) \ge \mathbb{E}(\alpha(G_{n,n,c/n}))$, since $|S|$ has the same distribution as $\alpha(G_{n,n,c/n})$.
\end{proof}

\subsection{The Regime of $c = \omega(1)$}

In this section we show that \textsc{Greedy} matches almost all offline nodes in $\RTPAM(n,c)$ model when $c = \omega(1)$. Consider the two-step view of $\RTPAM(n,c)$. Recall that $Z_j$ refers to the number of nodes of type $j$ that are generated, and that $X_j$ refers to the number of neighbors of a node of type $j$ that have not been matched in any of the previous rounds. Also, recall that \textsc{Greedy} matches $\mathbb{E}(\min(X_j, Z_j))$ in round $j$ in expectation. We will show that in most rounds $\min(X_j, Z_j)$ is very close to $Z_j$. This finishes the argument, since $\sum_{j}\mathbb{E}(Z_j) = n$ is the total number of offline nodes and a trivial upper bound on the size of a maximum matching.

\begin{theorem}
\label{thm:omega_1_c_regime}
Let $c = \omega(1)$ then we have:
\[\rho (\textsc{Greedy}, \RTPAM(n,c)) = 1.\]
\end{theorem}
\begin{proof}
Let $p=\frac{c}{n}$ and $k=\frac{n}{\sqrt{c}}$. Fix a round $i$ and assume that at least $k$ offline nodes have \emph{not} been matched in earlier rounds. Then variable $X_i$ has binomial distribution with at least $k$ trials and the probability of success $c/n$. We will consider $\widetilde{X}_i \sim \Bin(k, c/n)$ such that $\widetilde{X}_i \le X_i$. We will show that $\Pr(\widetilde{X}_i \ge Z_i) = 1 -o(1)$. Since $Z_i \sim \Poi(1)$ we have:
\[P(Z_i \geq c^{1/100}) = \frac{1}{e}\sum_{j=c^{1/100}}^\infty \frac{1}{j!} \leq \frac{1}{c^{1/100}!} \leq \frac{1}{2^{c^{1/100}}}. \]
For $\widetilde{X}_i$ we have $\Var(\widetilde{X}_i) = k p (1-p) = \sqrt{c}(1-c/n)$ and $\mathbb{E}(\widetilde{X}_i) = \sqrt{c}$. By Chebyshev's inequality
\[\Pr(|\widetilde{X}_i - \sqrt{c}| \ge c^{1/3}) \le \frac{\sqrt{c}(1-c/n)}{c^{2/3}} = \frac{1-c/n}{c^{1/6}}.\]
From these two bounds, it follows that
\[\Pr(\widetilde{X}_i\geq Z_i)\ge \Pr(Y_i \le c^{1/100}  \wedge  \widetilde{X}_i\ge \sqrt{c}-c^{1/3} )
	\ge 1-\frac{1-c/n}{c^{1/6}}-\frac{1}{2^{c^{1/100}}} = 1 - o(1).\]
In addition, it is easy to see that in the first $n-10k$ rounds \textsc{Greedy} matches at most $n-k$ offline nodes with probability $1-o(1)$. In particular, the probability of matching more than that is bounded by the probability that $\sum_{i=1}^{n-10k} Z_i > n-k$. Thus, we can condition on having at least $k$ available offline nodes during each of the first $n-10 k$ rounds. Therefore, the expected size of the matching constructed by \textsc{Greedy} is at least $\sum_{i=1}^{n-10k}\mathbb{E}(\min(Z_i, X_i)) \ge (n-10k)(1-o(1)) = n - o(n).$
	
	
\end{proof}

\subsection{Putting it together}

In this section, we take a closer look at our results for \textsc{Greedy} in $\RTPAM(n,c)$ model. We already know that \textsc{Greedy} achieves competitive ratio $1$ in the regimes $c = o(1)$ and $c = \omega(1)$. Hence, we concentrate on the regime of constant $c$. In Figure~\ref{fig:both-on-one} we plot the asymptotic fraction of matched offline nodes by \textsc{Greedy} (Theorem~\ref{thm:greedy_lb}) and the upper bound on the fraction of offline nodes in a maximum matching (Theorem~\ref{thm:match_ub}) as functions of $c$ .

By taking the ratio of the two curves in Figure~\ref{fig:both-on-one} we obtain a lower bound on the asymptotic approximation ratio of \textsc{Greedy} in the $\RTPAM(n,c)$ model as a function of $c$. We plot this lower bound in Figure~\ref{fig:approx_ratio}. We see that the lower bound achieves a unique minimum on the interval $(0, \infty)$ and that it converges to $1$ as $c$ goes to infinity. By numerically minimizing the lower bound we obtain that the minimum of this curve is achieved at $c \approx 0.667766$ and the lower bound is $\approx 0.715071$. Thus, we have the following corollary:

\begin{corollary}
For all regimes of $c$ we have
\[ \rho(\textsc{Greedy}, \RTPAM(n,c)) \ge 0.715.\]
\end{corollary}

The shape of the lower bound graph in Figure~\ref{fig:approx_ratio} is a bit strange and it suggests that our lower bound might not be tight. Therefore, we conjecture that it should be possible to strengthen the upper bound on the size of a maximum matching in $\RTPAM(n,c)$.

It is also interesting to compare the performance of \textsc{Greedy} on $\RTPAM(n,c)$ inputs with its performance on $G_{n,n,c/n}$ inputs. As stated in the introduction, we expect \textsc{Greedy} to perform worse in the $\RTPAM(n,c)$ model, because the $\RTPAM(n,c)$ model has ``less randomness'' in the sense of introducing correlations between consecutive online nodes that are not present in the $G_{n,n,c/n}$ model. Indeed, this intuition turns out to be correct. We plot the performance of \textsc{Greedy} in two models in Figure~\ref{fig:greedy_comp} and observe that \textsc{Greedy} on $G_{n,n,c/n}$ inputs outperforms $\RTPAM(n,c)$ for constant $c$.

%% file: conclusion.tex
We have introduced a new stochastic model for the  online bipartite matching problem. In our model, a random Erd\H{o}s-R\'{e}nyi type graph is generated first. Then an input instance graph is generated in rounds where in the $i^{th}$ round, the corresponding input type node appears consecutively $n_i$ times where $n_i$ is distributed according to the Poisson distribution with parameter 1. 

More generally, this model is just a specific case of a broad  class of stochastic online models for graph problems where type graphs are generated by some random or adversarial processes and then the $i^{th}$ input type node occurs consecutively $n_i$ times where $n_i$ is determined by another random process so as to model some limited form of ``locality of reference''. More generally, we could use a Markov process to generate the next type node instance. 

The $G_{n,n,p}$  Erd\H{o}s-R\'{e}nyi graphs (where $n_i = 1$ for all $i$) and   i.i.d. models
where the type graph $G$ is determined adversarially or according to a random  process (and where  the $i^{th}$ round is drawn i.i.d. from the type graph) fit within this general class of stochastic models. 

As in Mastin and Jaillet~\cite{MastJailGnnp} and Besser and  Poloczek~\cite{BesserP17}, we analyze the performance of the simplest greedy algorithm.  As in other such studies, it is often the case that simple greedy or ``greedy-like'' algorithms perform well on real benchmarks or stochastic settings, well beyond what worst case analysis might suggest.  Our specific $\RTPAM$ model introduces dependencies between online nodes that do not appear in other stochastic models for maximum bipartite matching.  
These dependencies in the $\RTPAM$ model  
result in some technical challenges in addition to the non-trivial analysis in Mastin and Jaillet. As in Mastin and Jaillet, our analysis falls into three classes dependening on the edge probabilities 
$p =  c/n$.
As in Mastin and Jaillet, the regimes $c = o(1)$ and $c = \omega(1)$ result in approximation ratios that approach 1 as $n$ increases.
And as in Mastin and Jaillet, we obtain an almost precise approximation ratio (modulo the estimate of the expected size of  optimal matching) for the regime of constant $c$. Given the input  dependencies our worst case bound (i.e. for the $c^*$ that minimizes the approximation ratio) is significantly less than in Mastin and Jaillet. 

As we have suggested, our $\RTPAM$ model is just a specific case of a wide class of online stochastic models that have not been studied with respect to any algorithm. We believe that such a study will be both technically interesting as well as becoming more applicable to many ``real-world'' settings where there is ``locality of reference''. We briefly discuss a further extension of our $\RTPAM$ model and its connection to the existing literature in Appendix~\ref{sec:3-step}. Finally, we have begun an experimental study of the performance of Greedy in comparison to algorithms that exploit the underlying  type graph in a distributional model (e.g.,  \cite{ FeldmanMMM2009,MGS,HMZ,Jaillet2014,BKP2018}).    

%% file: tech_lemmas.tex
In this appendix, we prove two lemmas that are used in Section~\ref{sec:greedy_poisson}. The first lemma gives a closed-form expression for $h(x) := \mathbb{E}(\min(\Poi(x), \Poi(1)))$ in terms of the modified Bessel functions of the first kind and the Marcum's Q functions. The derivation relies on the answer from Stats Stackexchange~\cite{Skellam}.

\begin{lemma}[Lemma~\ref{lem:h_expr} restated,\cite{Skellam}]
For $x > 0$ we have
\[ h(x) = \frac{1+x - 2 e^{-x-1}\left(I_0(2 \sqrt{x}) + \sqrt{x}I_1(2\sqrt{x})\right) - (1-x)\left(1-2Q_1(\sqrt{2 x}, \sqrt{2})\right)}{2}.\]
\end{lemma}
\begin{proof}
A random variable that is equal to the difference between two Poisson random variables has Skellam distribution. As the first step, we reduce the computation of $h(x)$ to the computation of the expectation of an absolute value of a Skellam distributed random variable:
\begin{align*}
    h(x) &= \mathbb{E}(\min(\Poi(x), \Poi(1))) = \frac{\mathbb{E}(\Poi(x)) + \mathbb{E}(\Poi(1)) - \mathbb{E}(|\Poi(x)-\Poi(1)|)}{2}  \\
    &= \frac{1 + x - \mathbb{E}(|\Poi(x)-\Poi(1)|)}{2}.
\end{align*} 
In the rest of the proof, we show how to compute $\mathbb{E}(|\Poi(x) - \Poi(1)|)$. From the PMF of a Skellam variable, one easily obtains the PMF of the absolute value of our Skellam variable:
\[    \Pr(|\Poi(x)-\Poi(1)|=k) = \left\{
\begin{array}{ll}
    e^{-1-x} \left(x^{k/2} I_k(2\sqrt{x})+x^{-k/2}I_{-k}(2\sqrt{x}) \right) & \text{ if $k > 0$}, \\
    e^{-1-x} I_0(2\sqrt{x}) & \text{ if $k = 0$}.
\end{array} 
\right.\]
Then we write down the MGF and simplify it using the Marcum's Q function to get:
\begin{align*}
M_{|\Poi(x)-\Poi(1)|}(t) = &e^{-1-x}\left(Q_1(\sqrt{2 \exp(-t)}, \sqrt{2 x \exp(t)})\exp(xe^t+e^{-t}) \right.\\
&\left. +Q_1(\sqrt{2x \exp(-t)}, \sqrt{2  \exp(t)})\exp(e^t+xe^{-t})-I_0(2\sqrt{x})\right).
\end{align*}
Now, we can take the derivative of the MGF at $t=0$ to derive:
\[ \mathbb{E}(|\Poi(x)-\Poi(1)|) = 2 e^{-x-1}\left(I_0(2 \sqrt{x}) + \sqrt{x}I_1(2\sqrt{x})\right) + (1-x)\left(1-2Q_1(\sqrt{2 x}, \sqrt{2})\right).\]
\end{proof}

In the next lemma we prove that the function defining the differential equation in Section~\ref{sec:greedy_poisson} is Lipschitz. This is needed to establish the main result of the paper via Wormald's theorem.

\begin{lemma}[Lemma~\ref{lem:key_lem_2} restated]
The function $f_c(x)$ is Lipschitz on $[0,1]$.
\end{lemma}
\begin{proof}
We prove the statement by showing that the derivative of $f_c(x)$ is bounded on $[0,1]$. By definition of $f_c$ we have $f_c(x) = h(c(1-x)).$ Thus, $f_c'(x) = -c h'(c(1-x))$. Hence bounding $|f_c'(x)|$ on $[0,1]$ amounts to bounding $h'(x)$ on $[0,c]$. We compute the derivative of $h$ as follows:
\begin{align*}
2 h'(x) &= 2 + 2 e^{-1-x}(I_0(2\sqrt{x})+\sqrt{x}I_1(2\sqrt{x}))-e^{-1-x}(3I_1(2\sqrt{x})/\sqrt{x}+I_0(2\sqrt{x})+I_2(2\sqrt{x}))\\
&\;\;\;\; -2Q_1(\sqrt{2x},\sqrt{2})+2(1-x)(Q_2(\sqrt{2x},\sqrt{2})-Q_1(\sqrt{2x},\sqrt{2})) \\
&= 2 + 2 e^{-1-x}(I_0(2\sqrt{x})+\sqrt{x}I_1(2\sqrt{x}))-e^{-1-x}(3I_1(2\sqrt{x})/\sqrt{x}+I_0(2\sqrt{x})+I_2(2\sqrt{x}))\\
&\;\;\;\; -2Q_1(\sqrt{2x},\sqrt{2})+2(1-x)e^{-1-x}I_1(2\sqrt{x})/\sqrt{x} \\
&= e^{-1-x}(I_0(2\sqrt{x})-I_1(2\sqrt{x})/\sqrt{x}-I_2(2\sqrt{x})) +2-2Q_1(\sqrt{2x},\sqrt{2}), 
\end{align*}
where the second equation follows from the definition of the Marcum's Q function and the third equation follows by collecting and simplifying terms with the factor of $e^{-1-x}$ together. We complete the proof of the lemma by bounding each of the terms.

By the definition of the Bessel functions of the first kind we have
\begin{align*}
I_0(2\sqrt{x})-I_1(2\sqrt{x})/\sqrt{x}-I_2(2\sqrt{x}) &= \sum_{i=0}^\infty \frac{1}{i! i!} x^i - \sum_{i=0}^\infty \frac{1}{i! (i+1)!} x^i - \sum_{i=1}^\infty \frac{1}{(i-1)!(i+1)!} x^i \\
&= 1 - 1 + \sum_{i=1}^\infty \left(\frac{1}{i!i!}-\frac{1}{i!(i+1)!}-\frac{1}{(i-1)!(i+1)!} \right) x^i\\
&= \sum_{i=1}^\infty \frac{1}{i!(i+1)!} x^i \le \sum_{i=0}^\infty \frac{1}{i!} x^i = e^x.
\end{align*}

Thus, the first term is bounded by $e^c$. The second term $2-2Q_1(\sqrt{2x},\sqrt{2})$ is bounded by $2$. This follows from interpretation of the Marcum's Q function as a probability --- in particular, we have $Q_1(\sqrt{2x}, \sqrt{2}) \in [0,1]$ (see~\cite{Helstrom68}). All in all, we have $|f_c'(x)| = O(1)$ for $x\in[0,1]$.
\end{proof}

%% file: figures.tex
In this appendix, we collect all figures mentioned in the paper.

\begin{figure}[ht]
\centering
\includegraphics[scale=0.25]{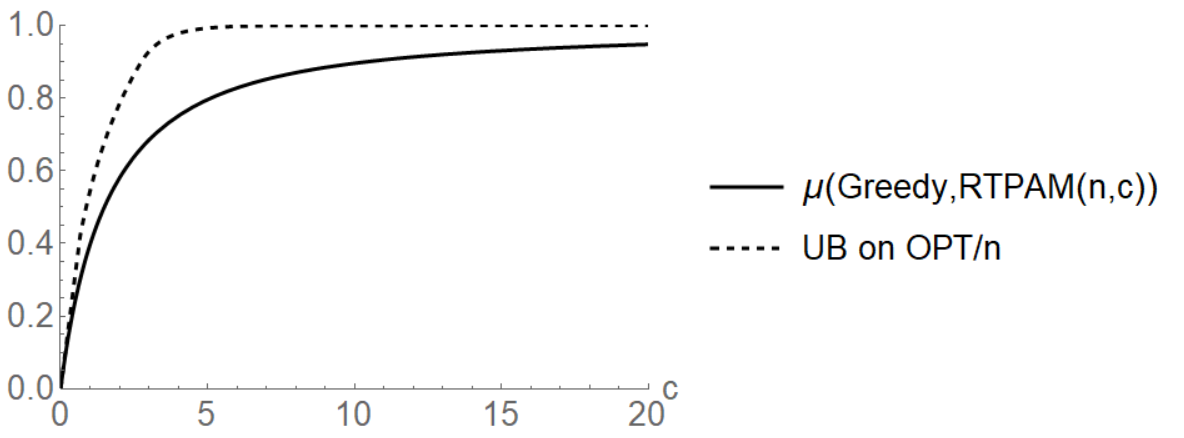}
\caption{The exact asymptotic ratio of \textsc{Greedy} to $n$ from Theorem~\ref{thm:greedy_lb}, and the upper bound on the asymptotic ratio of a size of maximum matching to $n$ from Theorem~\ref{thm:match_ub}. Both results are plotted as functions of $c$ in the $\RTPAM(n,c)$ model.}
\label{fig:both-on-one}
\end{figure}

\begin{figure}[ht]
\centering
\includegraphics[scale=0.25]{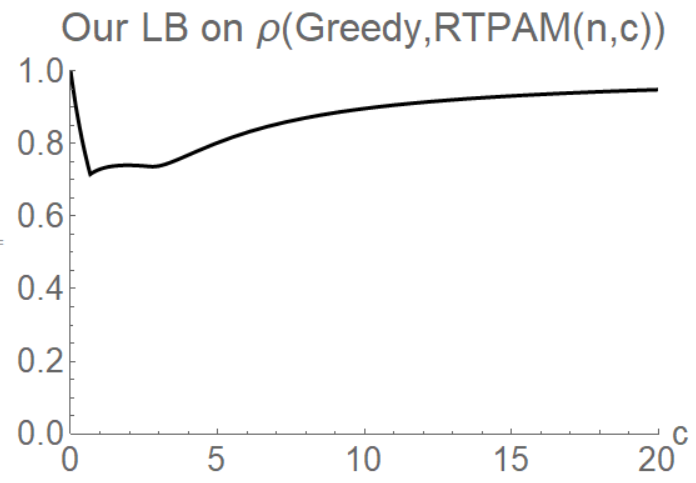}
\caption{The lower bound on the competitive ratio for \textsc{Greedy} in $\RTPAM(n,c)$ model as a function of $c$.}
\label{fig:approx_ratio}
\end{figure}

\begin{figure}[ht]
\centering
\includegraphics[scale=0.25]{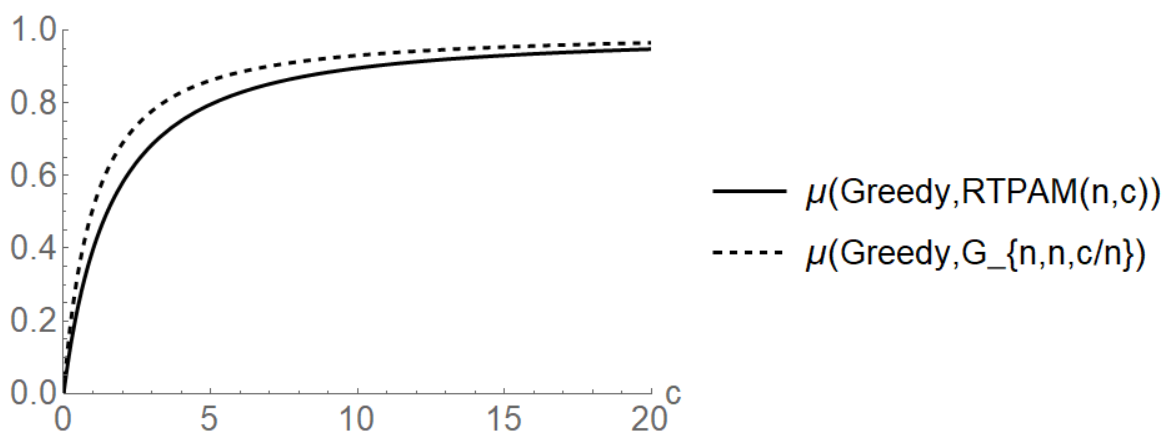}
\caption{The exact asymptotic ratio of \textsc{Greedy} to $n$ in $\RTPAM(n,c)$ input model ( Theorem~\ref{thm:greedy_lb}) versus the exact asymptotic ratio of \textsc{Greedy} to $n$ in $G_{n,n,c/n}$ input model (Theorem~\ref{thm:greedy_gnnp}).}
\label{fig:greedy_comp}
\end{figure}

\begin{figure}[ht]
\centering
\includegraphics[scale=0.25]{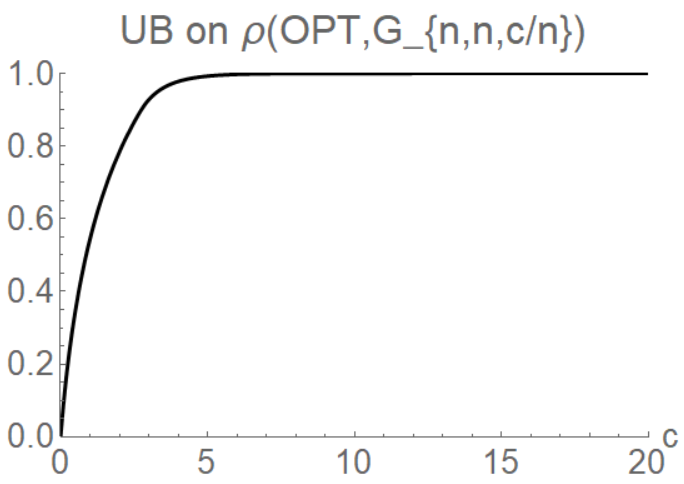}
\caption{The upper bound on the asymptotic ratio of the size of a maximum matching in $G_{n,n,c/n}$ to $n$ given in Theorem~\ref{thm:match_gnnp_ub}.}
\label{fig:match_gnnp_ub}
\end{figure}

%% file: three_step.tex
Another model of interest is defined by removing the assumption that online nodes corresponding to a given type occur consecutively. One can think about this model as the two-step $\RTPAM(n,c)$ model with an additional step of permuting online nodes. We call this model three-step $\RTPAM(n,c)$ and denote it by $3\mbox{-}\RTPAM(n,c)$. 

Recall that $Z_i$ is the number of online nodes of type $i$. In the conventional known i.i.d. model with integral types, we have that $Z_i \sim \Bin(n, 1/n)$ subject to $\sum_i Z_i = n$. The $3\mbox{-}\RTPAM(n,c)$ model is natural since $\Bin(n,1/n)$ converges to $\Poi(1)$ in distribution. Additionally, $3\mbox{-}\RTPAM(n,c)$ removes the dependency $Z_{1}+\cdots+Z_{n}=n$ and the total number of online nodes is $n$ in expectation, but can differ from $n$ on any particular random instantiation. In many proofs in the existing literature on the conventional known i.i.d. model, authors are essentially using the approximation $\Bin(n, 1/n) \approx \Poi(1)$ to conclude a proof (which explains why the number $1/e$ is so ubiquitous), while dealing with $\Bin(n,1/n)$ in intermediate computations. Working with $3\mbox{-}\RTPAM(n,c)$ directly seems more elegant. Moreover, Claim~\ref{claim:equiv} and the discussion following it show that one can transfer the results for $3\mbox{-}\RTPAM(n,c)$ to the conventional known i.i.d. model without any deterioration of parameters. Arguably, because of these reasons, the $3\mbox{-}\RTPAM(n,c)$ model is more natural than the conventional known i.i.d. model. 

It is easy to observe that the proof of Theorem~\ref{thm:little_c_regime} ($c = o(1)$) goes through for $3\mbox{-}\RTPAM(n,c)$. The same holds for the proof of Theorem~\ref{thm:omega_1_c_regime} ($c = \omega(1)$) because of the principle of delayed decisions. The case of constant $c$ is an open problem and a natural question to study next.

To our knowledge, the first mention of a Poisson distribution in online bipartite matching modelling was the Poisson arrivals model of Jaillet and Lu~\cite{Jaillet2014}. In that model, the first step is to sample the total number of online nodes $Z$ from $\Poi(n)$. The second step is to draw $Z$ online nodes i.i.d. from the known distribution. Jaillet and Lu were motivated by relaxing the assumption that the number of online nodes is known in advance exactly. It is easy to see that the Poisson arrivals model is equivalent to $3\mbox{-}\RTPAM(n,c)$. We present this calculation here for completeness.

\begin{claim}
\label{claim:equiv}
The $3\mbox{-}\RTPAM(n,c)$ model is equivalent to the Poisson arrivals model.
\end{claim}
\begin{proof}
We consider the expected number of online nodes of type $i$. In $3\mbox{-}\RTPAM(n,c)$, the probability of exactly $k$ occurrences of type $i$ is $\frac{1}{e\cdot k!}$. In the Poisson arrivals model, the probability of type $i$ occurring $k$ times is the following:

\begin{align*}
\Pr(\text{exactly }k\text{ occurrences of type }i) &= \sum_{j=k}^{\infty}\frac{n^{j}}{e^{n}j!}\binom{j}{k}\left(\frac{1}{n}\right)^{k}\left(1-\frac{1}{n}\right)^{j-k}\\
&= \sum_{j=k}^{\infty} \frac{1}{e^{n}}\cdot \frac{1}{k! (j-k)!}\cdot (n-1)^{j-k}\\
&= \frac{1}{e^{n}k!} \sum_{j=0}^{\infty} \frac{(n-1)^{j}}{j!} = \frac{e^{n-1}}{e^{n}k!} = \frac{1}{e\cdot k!}.
\end{align*}
Thus, the number of online nodes of a given type, i.e., the $Z_i$,  are distributed identically in the two models. Moreover, conditioned on the values of the $Z_i$, the order in which the online nodes appear in both models is distributed according to a random permutation.
\end{proof}

Finally, \cite{Jaillet2014} prove that a \emph{greedy} $d$-competitive algorithm in the conventional known i.i.d. model is $d$-competitive in the Poisson arrivals model and therefore in $3\mbox{-}\RTPAM(n,c)$, and vice versa.
